%% file: main.tex
\documentclass{multibody2023_full_paper}

\usepackage{import,color}

\usepackage{amsthm}

\newtheorem{theorem}{Theorem}[section]

\newtheorem{remark}[theorem]{Remark}
\newtheorem{definition}[theorem]{Definition}
\newtheorem{proposition}[theorem]{Proposition}

\newcommand{\dt}{\mathrm{d}t}
\renewcommand{\d}[1][]{\ensuremath{\,\mathrm{d}#1}}
\DeclareMathOperator{\D}{D}
\newcommand{\transp}{^\mathrm{T}}
\newcommand{\inv}{^{-1}}
\renewcommand{\vec}[1]{\ensuremath{\mbox{$\mathbf{#1}$}}}
\newcommand{\DG}{\overline{\mathsf{D}}}
\newcommand{\npe}{^{n+1}}
\newcommand{\npeh}{^{n+1/2}}
\newcommand{\n}{^{n}}
\newcommand{\eqcomma}{\ensuremath{\ \text{,}}}
\newcommand{\eqdot}{\ensuremath{\ \text{.}}}

\usepackage{pgfplots}
\pgfplotsset{compat = newest}
\usepackage{pgfplotstable}
\usepackage{tikz}
\pgfkeys{/pgf/number format/.cd,fixed,precision=4}
\newlength\figH
\newlength\figW

\input{colors.tex}

\usepackage{url}

\AddToHook{shipout/foreground}{
       \begin{tikzpicture}[overlay, remember picture]
            \node[fill=none, draw, text=red, rotate=0, scale=1, xshift=2.5cm, yshift=16.6cm] at (current page) {
				 Note that this article has an open-access published version at {\color{red} \url{www.doi.org/10.1007/s11044-024-10001-9}} .};
        \end{tikzpicture}    
    }

\begin{document}

\begin{center}
    \fontsize{14}{16}{\bf
        Energy-consistent integration of mechanical systems \\ \vspace{0.4ex}
        based on Livens principle
    }
\end{center}

\begin{center}
    \normalsize{
        \bf{
            \underline{Philipp L. Kinon},
            Peter Betsch}
    }
\end{center}

\begin{center}
    \begin{tabular}{c}
        Institute of Mechanics                        \\
        Karlsruhe Institute of Technology (KIT)       \\
        Otto-Ammann-Platz 9, 76131 Karlsruhe, Germany \\
        {[philipp.kinon, peter.betsch]@kit.edu  }     \\
    \end{tabular} \\
\end{center}


\begin{quote} 

    \section*{ABSTRACT}
    In this work we make use of Livens principle (sometimes also referred to as Hamilton-Pontryagin principle) in order to obtain a novel structure-preserving integrator for mechanical systems. In contrast to the canonical Hamiltonian equations of motion, the Euler-Lagrange equations pertaining to Livens principle circumvent the need to invert the mass matrix. This is an essential advantage with respect to singular mass matrices, which can yield severe difficulties for the modelling and simulation of multibody systems. Moreover, Livens principle unifies both Lagrangian and Hamiltonian viewpoints on mechanics. Additionally, the present framework avoids the need to set up the system's Hamiltonian. The novel scheme algorithmically conserves a general energy function and aims at the preservation of momentum maps corresponding to symmetries of the system. We present an extension to mechanical systems subject to holonomic constraints. The performance of the newly devised method is studied in representative examples.

    \textbf{Keywords:} Livens principle, Structure-preserving integration, Singular mass matrix, Energy-momentum method, Discrete gradients, Constrained mechanical systems.

\end{quote}

\section{INTRODUCTION}
Over time, two approaches to the description of dynamical systems have been developed: The well-known Lagrangian and Hamiltonian formalisms both consider descriptive energetic scalars and deploy certain operations on them to generate the system's equations of motion. However, another formulation, which unifies both frameworks by means of independent position, velocity and momentum quantities has been proposed by Livens \cite{livens_hamiltons_1919}. This Livens principle has been recently taken up (cf. \cite{bou-rabee_hamiltonpontryagin_2009,holm_geometric_2011}) under the name of Hamilton-Pontryagin principle due to its close relation to the Pontryagin principle from the field of optimal control \cite{pontryagin_mathematical_1962}. Livens principle allows for an advantageous universal description due to its mixed character.

Additionally, when simulating multibody system dynamics, the mathematical properties of the formulation heavily depend on the particular choice of coordinates. In certain frameworks, the mass matrix can even become singular or does depend on the coordinates themselves (cf. \cite{udwadia06, garcia_de_jalon13}). These cases pose some difficulties for Hamiltonian methods, since the inversion of the mass matrix might not be admissible.

Thus, in the present work, we propose the application of Livens principle to multibody system dynamics which circumvents the problem of singular mass-matrices. This extends previous works \cite{kinon_ggl_2023, kinon_structure_2023}, in which a novel variational principle has been derived based on Livens principle, with respect to more general mass matrices.

In recent times, the focus of scientific research has become to derive numerical integration methods, which are capable of solving the equations of motion approximately. Therefore, the class of structure-preserving integrators seeks to inherit the conservation principles of dynamical systems in a discrete sense (cf. monographs such as Hairer et al. \cite{hairer_geometric_2006}). In the field of mechanics, structure-preserving integration schemes can be mainly divided into two different groups: symplectic methods, of which variational integrators are an important subclass, and energy-momentum integrators.

In the present contribution we focus on the approach of energy-momentum methods: Based on the direct discretization of the Livens equations of motion in conjunction with Gonzalez discrete derivatives \cite{gonzalez_time_1996}, we obtain an energy-consistent time-stepping scheme. This method is able to take into account both singular and configuration-dependent mass matrices. It achieves discrete energy preservation and satisfies holonomic constraint functions exactly.

The outline of the present work is as follows: In Sec.~\ref{sec_ham} the fundamentals of Hamiltonian dynamics are briefly recapitulated. Thereafter, we introduce Livens principle in Sec.~\ref{sec_livens} and extend it to holonomically constrained systems in Sec.~\ref{sec_constraints}. In Sec.~\ref{sec_discrete} we discretize the deduced equations of motion to obtain an energy-consistent time-stepping scheme, which is used in Sec.~\ref{sec_examples} for the simulation of representative example problems. Sec.~\ref{sec_conclusion} concludes the findings and gives a short outlook for future research.

\section{HAMILTONIAN DYNAMICS}\label{sec_ham}
Consider the motion of a dynamical system with $d$ degrees of freedom with positions $\vec{q} \in Q$ and velocities $\dot{\vec{q}} \in T_\mathbf{q} Q$, where $ Q \subset \mathbb{R}^d$ denotes the configuration space and $T_\mathbf{q} Q \subset \mathbb{R}^d$ the tangent space to $Q$ at $\vec{q}$.
\begin{definition}
    The system's Lagrangian $L \colon TQ \to \mathbb{R}$ is given by
    \begin{align}
        L(\vec{q},\dot{\vec{q}}) = T(\vec{q},\dot{\vec{q}}) - V(\vec{q}) = \frac{1}{2} \dot{\vec{q}} \cdot \vec{M}(\vec{q}) \dot{\vec{q}} -V(\vec{q}), \label{lagrangefunction}
    \end{align}
    where $T:TQ \rightarrow \mathbb{R}$ is the kinetic energy, $\vec{M}(\vec{q}) \in \mathbb{R}^{d \times d}$ denotes the symmetric and positive-definite mass matrix and $V: Q \rightarrow \mathbb{R}$ is a potential function.
\end{definition}

The total energy of the system is given by the sum of kinetic and potential energy, i.e.
\begin{align} \label{total_energy}
    E_\mathrm{tot}(\vec{q},\vec{v}) = T(\vec{q},\vec{v}) + V(\vec{q}) .
\end{align}

\begin{definition}
    The Hamiltonian $H: T^*Q \rightarrow \mathbb{R} $ is obtained by employing a Legendre transformation $\mathbb{F}_L: (\vec{q},\dot{\vec{q}}) \mapsto (\vec{q},\vec{p})$, where the conjugate momenta are defined throught the fiber derivative
    \begin{align}
        \vec{p}:=D_2 L(\vec{q}, \dot{\vec{q}})
    \end{align}
    such that
    \begin{align}
        H(\vec{q},\vec{p}) =  \vec{p} \cdot \dot{\vec{q}}(\vec{q},\vec{p}) - L(\vec{q},\dot{\vec{q}}(\vec{q},\vec{p})) . \label{legendre}
    \end{align}
\end{definition}
Given a Lagrangian \eqref{lagrangefunction}, the Hamiltonian reads
\begin{align}
    H(\vec{q},\vec{p}) = T(\vec{q},\vec{p}) + V(\vec{q}) = \frac{1}{2} \vec{p} \cdot \vec{M}(\vec{q})^{-1} \vec{p} + V(\vec{q}) .\label{hamiltonian}
\end{align}
Eventually, the well-known Hamiltonian equations of motion appear in their canonical form
\begin{subequations}
    \label{canonical_Ham_EL}
    \begin{align}
        \dot{\vec{q}} & = \D_2 H(\vec{q},\vec{p})   \\
        \dot{\vec{p}} & = - \D_1 H(\vec{q},\vec{p})
    \end{align}
\end{subequations}
which can be specified for the systems with Lagrangian \eqref{lagrangefunction} as
\begin{subequations}
    \label{canonical_Ham_EL_2}
    \begin{align}
        \dot{\vec{q}} & = \vec{M}(\vec{q})\inv \vec{p},  \label{canonical_Ham_EL_2-1}             \\
        \dot{\vec{p}} & = - \D_1 T(\vec{q},\vec{p}) - \D V(\vec{q}), \label{canonical_Ham_EL_2-2}
    \end{align}
\end{subequations}

Note that this approach requires to set up the Hamiltonian which in turn requires the inversion of the mass matrix. Thus, we now want to introduce a different and more general framework.

\section{LIVENS PRINCIPLE}\label{sec_livens}
From Hamilton's principle of least action
\begin{align}
    \delta \int_0^T L(\vec{q},\dot{\vec{q}}) \dt = 0
\end{align}
one can proceed by allowing the velocities to be independent variables $\vec{v} \in T_\mathbf{q} Q$. The kinematic relation $\dot{\vec{q}} = \vec{v}$ can be enforced by means of Lagrange multipliers $\vec{p}  \in T^*_\mathbf{q} Q$.

\begin{definition}
    The augmented action integral for Livens principle reads
    \begin{align}
        S(\vec{q},\vec{v},\vec{p}) = \int_0^T \left[ L(\vec{q},\vec{v}) + \vec{p} \cdot(\dot{\vec{q}} - \vec{v}) \right] \, \dt \label{livens_principle}
    \end{align}
    for the simulation time $t \in [0, T]$.
\end{definition}
By stating the stationary condition
\begin{align}
    \delta S( \vec{q},\vec{v}, \vec{p}) = 0
\end{align}
and executing the variations with respect to every independent variable, one obtains the equations of motion
\begin{subequations}
    \label{HP_EL}
    \begin{align}
        \dot{\vec{q}} & = \vec{v} \label{HP_EL1} ,                  \\
        \dot{\vec{p}} & = \D_1 L(\vec{q},\vec{v})  \label{HP_EL2} , \\
        \vec{p}       & = \D_2 L(\vec{q},\vec{v}) \label{HP_EL3},
    \end{align}
\end{subequations}
where the arbitraryness of variations and standard endpoint conditions $\delta \vec{q}(t=0) = \delta \vec{q}(t=T) = \vec{0}$ have been taken into account.

With regard to \eqref{HP_EL3} the Lagrange multiplier $\vec{p}$ can be identified as the conjugate momentum. Thus, Livens principle automatically accounts for the Legendre transformation \eqref{legendre}, whereas within the framework of Hamiltonian dynamics momentum variables have to be defined a priori using the fiber derivative. In analogy to \eqref{legendre}, a generalized energy function can be introduced as follows.
\begin{definition}
    Given a Lagrangian \eqref{lagrangefunction} the quantity
    \begin{align}
        E(\vec{q},\vec{v},\vec{p}) = \vec{p} \cdot \vec{v} - L(\vec{q},\vec{v}) \label{generalized_energy}
    \end{align}
    is called a generalized energy function.
\end{definition}

\begin{remark}[Alternative form of Livens principle]
    Using the generalized energy function, \eqref{livens_principle} can be recast as
    \begin{align}
        S(\vec{q},\vec{v},\vec{p}) = \int_0^T \left[ \vec{p} \cdot \dot{\vec{q}} - E(\vec{q},\vec{v},\vec{p}) \right] \dt ,
    \end{align}
    such that the corresponding Euler-Lagrange equations \eqref{HP_EL} emerge as
    \begin{subequations}
        \label{HP_EL_new}
        \begin{align}
            \dot{\vec{q}} & = \D_3 E(\vec{q},\vec{v},\vec{p}) ,  \\
            \dot{\vec{p}} & = -\D_1 E(\vec{q},\vec{v},\vec{p}) , \\
            \vec{0}       & = \D_2 E(\vec{q},\vec{v},\vec{p}) .
        \end{align}
    \end{subequations}
\end{remark}

\begin{proposition}
    The generalized energy \eqref{generalized_energy} is a conserved quantity along solutions of \eqref{HP_EL}.
\end{proposition}

\begin{proof}
    We compute
    \begin{align}
        \frac{\d}{\dt}E(\vec{q},\vec{v},\vec{p}) & = \D_1 E(\vec{q},\vec{v},\vec{p}) \cdot \dot{\vec{q}} + \D_2 E(\vec{q},\vec{v},\vec{p}) \cdot \dot{\vec{v}} + \D_3 E(\vec{q},\vec{v},\vec{p}) \cdot \dot{\vec{p}} \\
                                                 & = -\D_1 L(\vec{q},\vec{v}) \cdot \vec{v} + (\vec{p} - \D_2 L(\vec{q},\vec{v})) \cdot \dot{\vec{v}} + \vec{v} \cdot \D_1 L(\vec{q},\vec{v}) ,
    \end{align}
    where the partial derivatives of $E$ have been used and \eqref{HP_EL} has been inserted. Thus, the second term vanishes while the first and third term cancel each other out, such that $\d E / \dt = 0$.
\end{proof}

\begin{remark}[Equivalence with Lagrangian dynamics]
    Note that after reinserting \eqref{HP_EL3} into \eqref{HP_EL2} and making use of \eqref{HP_EL1}, Livens principle traces back to the Lagrangian equations of the second kind
    \begin{align}
        \frac{\d}{\dt} \left( \D_2 L(\vec{q},\dot{\vec{q}}) \right) - \D_1 L(\vec{q},\dot{\vec{q}}) = \vec{0} .
    \end{align}
\end{remark}

\begin{remark}[Equivalence with Hamiltonian dynamics]
    For mechanical systems with Lagrangian \eqref{lagrangefunction}, relation \eqref{HP_EL3} yields $\vec{p}=M(\vec{q})\vec{v}$, so that, if $M$ is nonsingular, \eqref{HP_EL1} and \eqref{HP_EL2} directly lead to the canonical Hamiltonian equations of motion \eqref{canonical_Ham_EL_2}. Moreover, in this case, the generalized energy function can be identified as the Hamiltonian, such that $E(\vec{q},\vec{v}(\vec{q},\vec{p}),\vec{p}) = H(\vec{q},\vec{p})$.
\end{remark}

Initially termed Livens principle (cf.\ Sec.\ 26.2 in Pars \cite{pars_treatise_1966}) after G.H. Livens, who proposed this functional for the first time (cf.\ Livens \cite{livens_hamiltons_1919}), Marsden and co-workers (e.g. \cite{bou-rabee_hamiltonpontryagin_2009}) coined the name \textit{Hamilton-Pontryagin principle} for this functional due to its close relation to the classical \textit{Pontryagin principle} from the field of optimal control. Due to its mixed character with three independent fields ($q,v,p$), it resembles the \textit{Hu-Washizu principle} from the area of elasticity theory. More recently, preliminary works \cite{kinon_ggl_2023, kinon_structure_2023} by the authors of this constribution have taken up Livens principle in order to obtain a novel variational principle, the \textit{GGL principle} which extends Livens principle to holonomically constrained systems. In these works, new structure-preserving integrators for systems with constant mass matrix have been developed.

\section{HOLONOMICALLY CONSTRAINED SYSTEMS}\label{sec_constraints}

Assume that the coordinates $\vec{q}$ are redundant due to the presence of $m$ independent scleronomic, holonomic constraints.
Arranging the constraint functions $g_k \colon \mathbb{R}^d \rightarrow \mathbb{R}$ ($k=1,\ldots,m$) in a column vector $\vec{g} \in \mathbb{R}^m$, the constraints are given by
\begin{align} \label{eq:constraints-holonomic}
    \vec{g}(\vec{q}) = \vec{0}  \eqdot
\end{align}
The constraint functions are assumed to be independent, so that the constraint Jacobian $ \D \vec{g}(\vec{q})$ has rank $m$. Note that the constraints give rise to the configuration manifold
\begin{equation}
    Q = \left\{ \vec{q} \in \mathbb{R}^d \;\vert\; \vec{g}(\vec{q})=\vec{0} \right\} \eqdot
\end{equation}

Correspondingly, Livens principle can be augmented to account for the constraints \eqref{eq:constraints-holonomic}. Accordingly, the functional assumes the form
\begin{equation} \label{eq:functional_Livens_holonom}
    \hat{S}(\vec{q},\vec{v},\vec{p},{\bm \lambda}) = {S}(\vec{q},\vec{v},\vec{p}) + \int\nolimits_0^T {\bm \lambda}\cdot\vec{g}(\vec{q})\, \dt \eqcomma
\end{equation}
where ${\bm \lambda}\in\mathbb{R}^m$ represents Lagrange multipliers for the enforcement of the constraints.
Imposing the stationary conditions on functional \eqref{eq:functional_Livens_holonom}, the resulting Euler-Lagrange equations are obtained as
\begin{subequations}
    \label{HP_EL_cons}
    \begin{align}
        \dot{\vec{q}} & = \vec{v} \eqcomma \label{HP_EL_cons1}                                                           \\
        \dot{\vec{p}} & = \D_1 L(\vec{q},\vec{v}) - \D \vec{g}(\vec{q})\transp {\bm \lambda}\eqcomma \label{HP_EL_cons2} \\
        \vec{p}       & = \D_2 L(\vec{q},\vec{v}) \label{HP_EL_cons3} \eqcomma                                           \\
        \vec{0}       & = \vec{g}(\vec{q})  \eqdot \label{HP_EL_cons4}
    \end{align}
\end{subequations}

\begin{remark}
    Due to the constraints, coordinates become redundant. Hence, the mass matrix $\vec{M}(\vec{q})$ contained in the Lagrangian can be singular such that the mass matrix is positive semidefinite in general.
\end{remark}
Note that \eqref{eq:constraints-holonomic} at every time induces the secondary constraints or hidden velocity constraints $\d \vec{g} / \dt = \vec{0}$. Taking into account \eqref{HP_EL_cons1}, one obtains
\begin{align}
    \D\vec{g}(\vec{q})\vec{v} = \vec{0} \label{secondary_constraints}
\end{align}
It is well-known that the above differential algebraic equations (DAEs) \eqref{HP_EL_cons} have differentiation index 3 (see, for example, \cite{kunkel_differential-algebraic_2006}), which can lead to numerical instabilities. An index reduction using the classical Gear-Gupta-Leimkuhler stabilization (\cite{gear_automatic_1985}) or an expansion of \eqref{eq:functional_Livens_holonom} to account also for the hidden constraints (cf. \textit{GGL principle} in \cite{kinon_ggl_2023, kinon_structure_2023}) can circumvent the arising problems.

We now want to demonstrate that the conservation principle also holds in the constrained setting.

\begin{proposition}
    The generalized energy \eqref{generalized_energy} is a conserved quantity along solutions of \eqref{HP_EL_cons}.
\end{proposition}

\begin{proof}
    As above, we compute
    \begin{align}
        \frac{\d}{\dt}E(\vec{q},\vec{v},\vec{p}) & = \D_1 E(\vec{q},\vec{v},\vec{p}) \cdot \dot{\vec{q}} + \D_2 E(\vec{q},\vec{v},\vec{p}) \cdot \dot{\vec{v}} + \D_3 E(\vec{q},\vec{v},\vec{p}) \cdot \dot{\vec{p}}                        \\
                                                 & = -\D_1 L(\vec{q},\vec{v}) \cdot \vec{v} + (\vec{p} - \D_2 L(\vec{q},\vec{v})) \cdot \dot{\vec{v}} + \vec{v} \cdot ( \D_1 L(\vec{q},\vec{v}) - \vec{G}(\vec{q})\transp {\bm \lambda} ) ,
    \end{align}
    where the partial derivatives of $E$ have been used and \eqref{HP_EL} has been inserted. Thus, the second term vanishes while the first and third term cancel each other out. Eventually,
    \begin{align}
        \frac{\d}{\dt}E(\vec{q},\vec{v},\vec{p}) & = {\bm \lambda} \cdot (\D\vec{g}(\vec{q})\vec{v})
    \end{align}
    remains. Having in mind that \eqref{secondary_constraints} holds, also this last part vanished, i.e. $\d E / \dt = 0$.
\end{proof}

\section{STRUCTURE-PRESERVING DISCRETIZATION}\label{sec_discrete}
In this work, a novel integration method is proposed, which conserves first integrals of the equations of motion at hand, e.g. the generalized energy function $E$. This scheme results from a direct discretization of the Euler-Lagrange equations \eqref{HP_EL_cons} emanating from constrained Livens principle. Particularly, discrete derivatives in the sense of Gonzalez \cite{gonzalez_time_1996} are taken into account and are defined as follows.

\begin{definition}[Gonzalez discrete gradient]
    For a given function $f : \mathbb{R}^m \rightarrow \mathbb{R} $ the discrete derivative $ \DG : \mathbb{R}^m \times \mathbb{R}^m \rightarrow \mathbb{R}^m $ according to \cite{gonzalez_time_1996} is defined by
    \begin{align} \label{eq:DD-Gonzalez}
        \DG f(\vec{x},\vec{y}) = \D f(\vec{z}) + \frac{f(\vec{y}) - f(\vec{x}) -  \D f(
            \vec{z}) \cdot \vec{v}  }{ ||\vec{v}||^2 } \vec{v}
    \end{align}
    with $\vec{z} = \frac{1}{2}(\vec{x} +\vec{y})$, $\vec{v} = \vec{y} - \vec{x}$ and $|| \bullet ||$ denoting the standard Euclidean norm on $\mathbb{R}^m$.
\end{definition}

Those discrete derivatives represent second-order approximation to the exact gradients such that discrete derivatives satisfy the directionality condition
\begin{align}
    \DG f(\vec{x},\vec{y}) \cdot \vec{v} = f(\vec{y}) -f (\vec{x}) \label{directionality}
\end{align}
as well as the consistency condition $\DG f(\vec{x},\vec{y}) = \D f(\vec{z}) + O(||\vec{v}||)$. Eventually, we propose the energy-consistent integrator, governed by
\begin{subequations}
    \label{HP_EL_discrete}
    \begin{align}
        \vec{q}\npe - \vec{q}\n & = h \, \vec{v}\npeh \label{HP_EL1_discrete} ,                                                           \\
        \vec{p}\npe - \vec{p}\n & = h \DG_1 L(\vec{q},\vec{v})  \label{HP_EL2_discrete} - h\, \sum_{k=1}^{m} \lambda_k \DG g_k(\vec{q}) , \\
        \vec{p}\npeh            & = \DG_2 L(\vec{q},\vec{v}) \label{HP_EL3_discrete} ,                                                    \\
        \vec{g}(\vec{q}\npe)    & = \vec{0} \label{HP_EL4_discrete} ,
    \end{align}
\end{subequations}
where $h$ denotes the time step size, $\vec{v}\npeh = \frac{1}{2}(\vec{v}\n + \vec{v}\npe)$ and $\vec{p}\npeh = \frac{1}{2}(\vec{p}\n + \vec{p}\npe)$. Moreover, $\DG_i$ denotes the partitioned discrete derivative with respect to the $i$-th argument. Directionality condition \eqref{directionality} extends to partitioned versions such that
\begin{align}
    \DG_1 T(\vec{q},\vec{v}) \cdot \Delta \vec{q} + \DG_2 T(\vec{q},\vec{v}) \cdot \Delta \vec{v} = T(\vec{q}\npe,\vec{v}\npe) - T(\vec{q}\n,\vec{v}\n) ,
\end{align}
where $\Delta \vec{q} = \vec{q}\npe - \vec{q}\n$ and $\Delta \vec{v} = \vec{v}\npe - \vec{v}\n$.

\begin{remark}
    Assuming \eqref{lagrangefunction}, the present method \eqref{HP_EL_discrete} can be specified as
    \begin{subequations}
        \label{HP_EL_discrete_2}
        \begin{align}
            \vec{q}\npe - \vec{q}\n & = h \, \vec{v}\npeh \label{HP_EL1_discrete_2} ,                                                                                    \\
            \vec{p}\npe - \vec{p}\n & = h \, \DG_1 T(\vec{q},\vec{v}) - h \, \DG V(\vec{q}) - h\, \sum_{k=1}^{m} \lambda_k \DG g_k(\vec{q})  \label{HP_EL2_discrete_2} , \\
            \vec{p}\npeh            & = \DG_2 T(\vec{q},\vec{v}) \label{HP_EL3_discrete_2} ,                                                                             \\
            \vec{g}(\vec{q}\npe)    & = \vec{0} .
        \end{align}
    \end{subequations}
\end{remark}

\begin{definition}
    A discrete version of the energy function \eqref{generalized_energy} is given by
    \begin{align} \label{discrete_energy_function}
        E\n & = \vec{p}\n \cdot \vec{v}\n - \frac{1}{2}\vec{v}\n \cdot \vec{M}(\vec{q}\n) \vec{v}\n + V(\vec{q}\n) .
    \end{align}
\end{definition}

\begin{proposition}
    The generalized energy \eqref{discrete_energy_function} is conserved by time-stepping method \eqref{HP_EL_discrete} in a discrete sense.
\end{proposition}

\begin{proof}
    Scalar multiplying \eqref{HP_EL2_discrete} with $\Delta \vec{q}$ and adding the scalar product of \eqref{HP_EL3_discrete} with $\Delta \vec{v}$ yields
    \begin{align}
        (\vec{p}\npe -\vec{p}\n) \cdot (\vec{q}\npe - \vec{q}\n) + h \, \vec{p}\npeh \cdot (\vec{v}\npe - \vec{v}\n)                                                           & =                                                                                                                  \\
        h \, \DG_1 L(\vec{q},\vec{v}) \cdot                                                                           (\vec{q}\npe - \vec{q}\n) - h\, \sum_{k=1}^{m} \lambda_k & \DG g_k(\vec{q}) \cdot (\vec{q}\npe - \vec{q}\n) + \DG_2 L(\vec{q},\vec{v}) \cdot (\vec{v}\npe - \vec{v}\n) \notag
    \end{align}
    Inserting \eqref{HP_EL1_discrete} and making use of the directionality property gives
    \begin{align}
        (\vec{p}\npe -\vec{p}\n) \cdot \vec{v}\npeh + \vec{p}\npeh \cdot (\vec{v}\npe - \vec{v}\n) & = L(\vec{q}\npe,\vec{v}\npe) - L(\vec{q}\n,\vec{v}\n)                                          \\
                                                                                                   & \qquad \qquad - \sum_{k=1}^{m} \lambda_k \left(g_k(\vec{q}\npe)-g_k(\vec{q}\n)\right) . \notag
    \end{align}
    After cancelling some terms on the left-hand side of the previous relation and taking into account that the constraints are identically satisfied in each time step (see \eqref{HP_EL4_discrete}), one eventually arrives at
    \begin{align}
        E(\vec{q}\npe,\vec{v}\npe, \vec{p}\npe) = E(\vec{q}\n,\vec{v}\n, \vec{p}\n)
    \end{align}
    which concludes the proof.
\end{proof}


\begin{proposition}
    The discrete fiber derivative \eqref{HP_EL3_discrete_2} in general does not imply a pointwise fulfillment of the continuous fiber derivative. That is, $\vec{p}\n = \vec{M}(\vec{q}\n)\vec{ v}\n$ does not hold in general. However, the relation $\vec{p}\n = \vec{M}\vec{ v}\n$ is satisfied for constant mass matrices.
\end{proposition}
\begin{proof}
    In the case of a constant mass matrix, the discrete derivative is given by $\DG_2 T (\vec{q},\vec{v}) = \vec{M} \vec{v}\npeh$. Accordingly, the discrete fiber derivative \eqref{HP_EL3_discrete_2} yields $\vec{p}\npeh = \vec{M} \vec{v}\npeh$. Consequently, provided that $\vec{p}\n = \vec{M} \vec{v}\n$ is satisfied, $\vec{p}\npe = \vec{M} \vec{v}\npe$ is fulfilled, too.
\end{proof}


\begin{remark}
    Using the proposed integrator \eqref{HP_EL_discrete_2}, the discrete total energy
    \begin{align}
        E_\mathrm{tot}\n = \frac{1}{2}\vec{v}\n \cdot \vec{M}(\vec{q}\n) \vec{v}\n + V(\vec{q}\n) \label{discrete_total_energy}
    \end{align}
    is in general not identical to the discrete generalized energy function \eqref{discrete_energy_function}. This is only the case, if $\vec{p}\n = \vec{M}(\vec{q}\n)\vec{ v}\n$.
\end{remark}



\section{NUMERICAL EXAMPLE}\label{sec_examples}
In the following, the newly proposed scheme \eqref{HP_EL_discrete_2} is applied to some numerical examples. Since this involves the solution of an implicit set of equations, Newton's method is used in every time step with a tolerance of $\epsilon_\mathrm{Newton}$. The computations have been performed using the code package available at \cite{kinon_metis_2023}, which can also be used for verification.

\subsection{Mass-spring system with redundant coordinates}

\begin{figure}[b!]
    \centering
    \def\svgwidth{0.6\textwidth}
    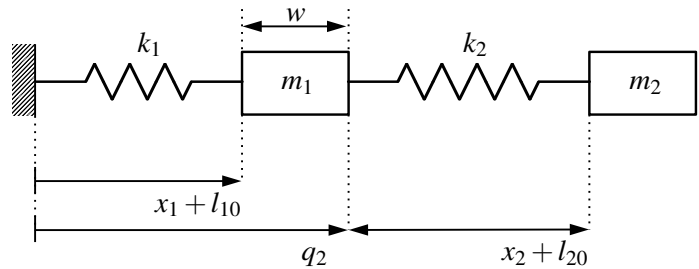
    \caption{Spring-mass multibody system.}
    \label{fig:federmasse}

\end{figure}

We take up an example with singular mass matrix from \cite{udwadia06} (cf. Sec.~5, Example 3) as depicted in Fig.~\ref{fig:federmasse} with masses $m_1$ and $m_2$ and springs with constants $k_1$ and $k_2$ and resting lengths $l_{10}$ and $l_{20}$, respectively. By splitting up the system into two seperate subsystems it can be considered a multibody system with two degrees of freedom. However, we use two absolute coordinates ($x_1$ and $q_2$) and one relative coordinate $x_2$ to define the systems kinematics, i.e. $d=3$, and
\begin{align}
    \vec{q}= \begin{bmatrix}
                 x_1 \\ q_2 \\ x_2
             \end{bmatrix}
\end{align}
with corresponding velocities
\begin{align}
    \vec{v} = \begin{bmatrix}
                  v_1 \\ v_2 \\ v_3
              \end{bmatrix} = \begin{bmatrix}
                                  \dot{x}_1 \\ \dot{q}_2 \\ \dot{x}_2
                              \end{bmatrix}.
\end{align}
Consequently, the interconnection constraint $ q_2 = x_1 + l_{10} + w $ arises, which can be recast as
\begin{align}
    g(\vec{q}) = \frac{1}{2} \left( (q_2 -x_1)^2 - (l_{10}+w)^2 \right) = 0 , \label{mass_spring_constraint}
\end{align}
such that $m=1$. The total kinetic energy is given by
\begin{align}
    T = \frac{1}{2} m_1 v_1^2 + \frac{1}{2} m_2 (v_2 + v_3)^2
\end{align}
and thus the mass matrix in \eqref{lagrangefunction} can be identified as
\begin{align}
    \vec{M} = \begin{bmatrix}
                  m_1 &  & 0   &  & 0   \\
                  0   &  & m_2 &  & m_2 \\
                  0   &  & m_2 &  & m_2
              \end{bmatrix}
\end{align}
and is singular. Thus, the inverse $\vec{M}\inv$ does not exist and we cannot find a Hamiltonian for this problem.
The potential is given by the elastic potential of the two springs, which is assumed to be nonlinear, such that
\begin{align}
    V(\vec{q}) = \frac{1}{4} k_1 (x_1^2 + x_1^4) +  \frac{1}{4} k_2 (x_2^2 + x_2^4) .
\end{align}
Since $\vec{M}$ is constant and does not dependent of the configuration,  the discrete generalized energy \eqref{discrete_energy_function} and the discrete total energy of the system \eqref{discrete_total_energy} are equivalent. Simulation parameters have been chosen as shown in Table~\ref{tab:tab1}.


\begin{table}[tb]
    \begin{center}
        \caption{Simulation parameters for mass-spring system with $i \in \{1,2\}$.}
        \vspace{1mm}
        \begin{tabular}{*{7}{|c}|}
            \hline
            $h$   & $T$  & $m_i$      & $k_i$        & $l_{i0}$     & $w$   & $\epsilon_\mathrm{Newton}$ \\
            \hline
            $0.1$ & $10$ & $\{1, 1\}$ & $\{ 1, 3 \}$ & $\{ 1, 1 \}$ & $0.1$ & $10^{-9}$                  \\
            \hline
        \end{tabular}
        \label{tab:tab1}
    \end{center}
\end{table}

The simulation yielded the energy evolutions displayed in Fig.~\ref{fig:massspring_energy}, where energy transfer between kinetic and potential energy becomes visible. The energy-consistent approximation provided by the proposed method becomes obvious in Fig.~\ref{fig:massspring_energy_diff}, where the temporal increment of the generalized energy is close to computer precision, such that it is a discretely conserved quantity. Furthermore, Fig.~\ref{fig:massspring_constraint} underlines the computationally exact treatment of the constraint equation in each time step.

\begin{figure}[b!]
    \centering
    \setlength{\figH}{0.18\textheight}
    \setlength{\figW}{0.6\textwidth}
    \input{figs/massspring/energy.tikz}
    \caption{Energy quantities}
    \label{fig:massspring_energy}
\end{figure}
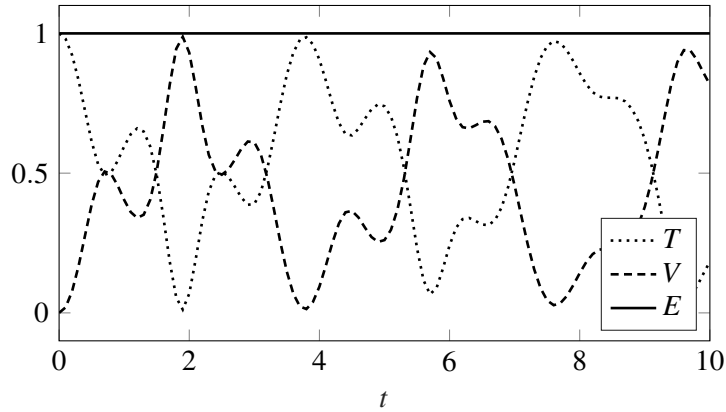

\begin{figure}[b!]
    \centering
    \hspace*{-10mm}
    \setlength{\figH}{0.18\textheight}
    \setlength{\figW}{0.6\textwidth}
    \input{figs/massspring/H_diff.tikz}
    \caption{Energy increments}
    \label{fig:massspring_energy_diff}
\end{figure}
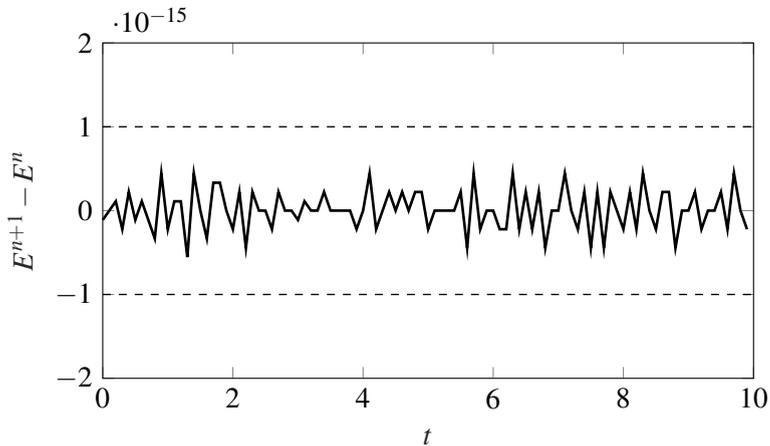

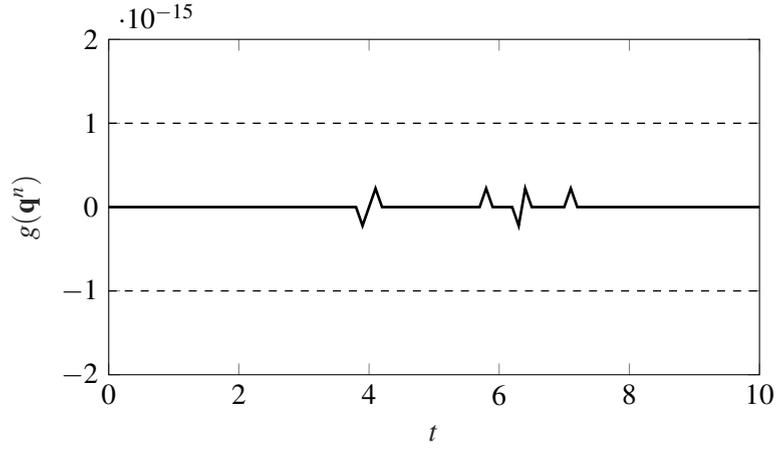
\begin{figure}[t!]
    \centering
    \setlength{\figH}{0.18\textheight}
    \setlength{\figW}{0.6\textwidth}
    \input{figs/massspring/g_pos.tikz}
    \caption{Positional constraint \eqref{mass_spring_constraint}}
    \label{fig:massspring_constraint}
\end{figure}

\subsection{Nonlinear spring pendulum with spherical coordinates}

\begin{figure}[b!]
    \centering
    \def\svgwidth{0.4\textwidth}
    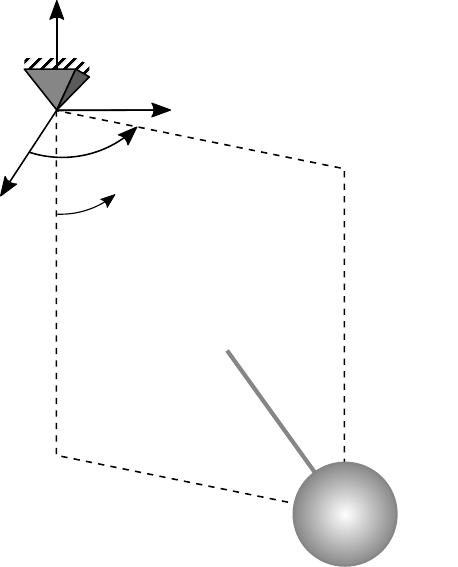
    \caption{Spring pendulum.}
    \label{fig:pendulum}
\end{figure}

We analyze the 3D spring pendulum with nonlinear spring potential formulated in spherical coordinates. The distance of mass $m$ from the Cartesian origin is given by $r \in \mathbb{R}_{\geq 0}$ and two angles $\theta, \varphi$ describe the orientation, as depicted in Fig.~\ref{fig:pendulum}. Consequently, the system has three degrees of freedom and
\begin{align}
    \vec{q} = \begin{bmatrix}
                  r \\ \theta \\ \varphi
              \end{bmatrix} .
\end{align}
The potential energy is given by the nonlinear internal potential with spring constant $EA$ such that
\begin{align}
    V(\vec{q}) = \frac{1}{2} EA \epsilon^2(r) ,
\end{align}
where the strain in the spring is computed via
\begin{align}
    \epsilon(r) = \frac{r^2 - l_0^2}{2 l_0^2}
\end{align}
and $l_0$ denotes the spring's resting length. The mass matrix from \eqref{lagrangefunction} for this example is given by
\begin{align}
    \vec{M}(\vec{q}) = m \, \mathrm{diag}(1, r^2, r^2 \sin^2(\theta))
\end{align}
and becomes singular if $\theta = k \pi$ for all integers $k \in \mathbb{Z}$.
Inital conditions have been chosen as
\begin{align}
    \vec{q}^0 = \begin{bmatrix}
                    1.05 \\ \pi/2 \\
                    0
                \end{bmatrix}, \quad \vec{v}^0 = \begin{bmatrix}
                                                     0 \\ 1 \\
                                                     1
                                                 \end{bmatrix}
\end{align}
to achieve an initial elongation of the spring by $5\%$ and a tangential initial velocity. Simulation parameters have been chosen as shown in Table~\ref{tab:tab2}.

\begin{table}[tb]
    \begin{center}
        \caption{Simulation parameters for spring pendulum.}
        \vspace{1mm}
        \begin{tabular}{*{6}{|c}|}
            \hline
            $h$    & $T$ & $m$ & $EA$  & $l_{0}$ & $\epsilon_\mathrm{Newton}$ \\
            \hline
            $0.01$ & $1$ & $1$ & $300$ & $1$     & $10^{-9}$                  \\
            \hline
        \end{tabular}
        \label{tab:tab2}
    \end{center}
\end{table}

In this example, the energetic quantities have been simulated as can be seen in Fig.~\ref{fig:springPendulum_energy}. While the conservation of the discrete total energy is captured down to an order of $10^{-4}$ (cf. Fig.~\ref{fig:springPendulum_total_energy_diff}), the generalized energy function is exactly preserved by the newly proposed scheme (see Fig.~\ref{fig:springPendulum_energy_function_diff}).


\begin{figure}[htb]
    \centering
    \setlength{\figH}{0.18\textheight}
    \setlength{\figW}{0.6\textwidth}
    \input{figs/springPendulum/energy.tikz}
    \caption{Energy quantities}
    \label{fig:springPendulum_energy}
\end{figure}
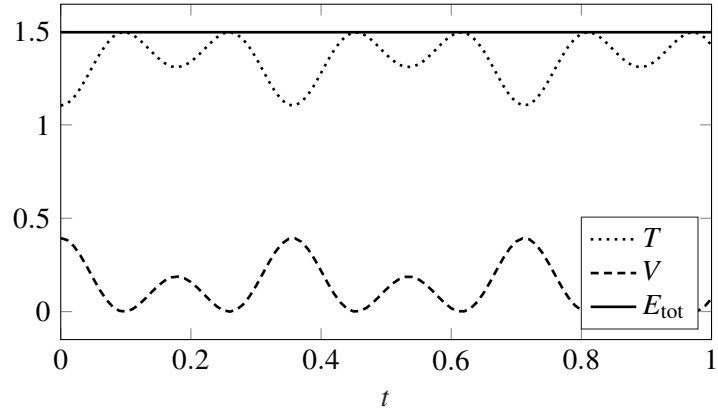

\begin{figure}[htb]
    \centering
    \hspace*{-10mm}
    \setlength{\figH}{0.18\textheight}
    \setlength{\figW}{0.6\textwidth}
    \input{figs/springPendulum/H_diff.tikz}
    \caption{Discrete total energy \eqref{discrete_total_energy} increments}
    \label{fig:springPendulum_total_energy_diff}
\end{figure}
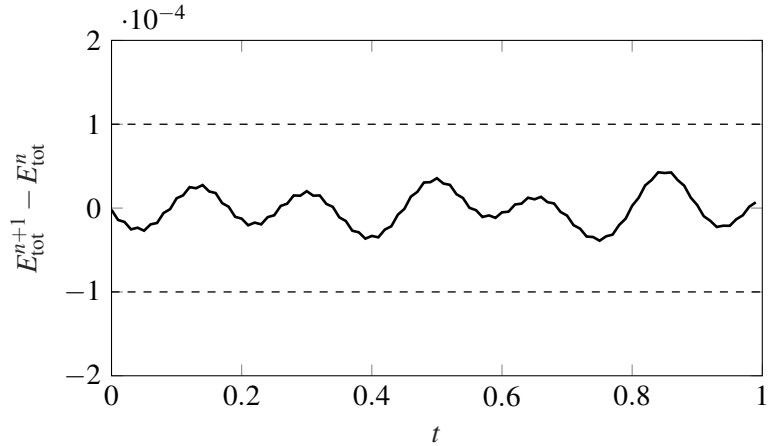

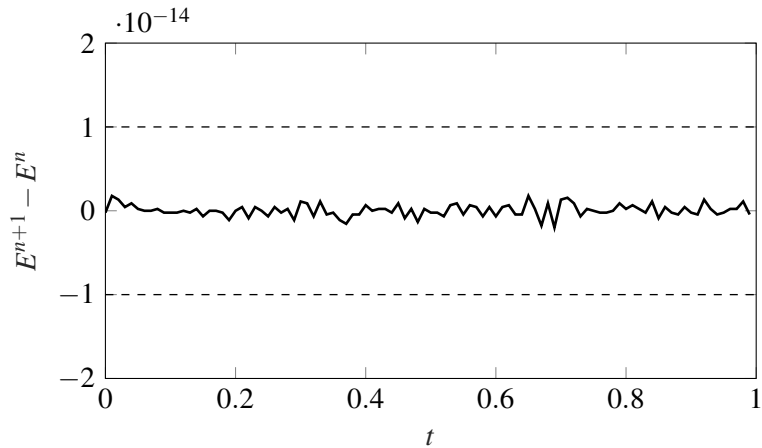
\begin{figure}[htb]
    \centering
    \setlength{\figH}{0.18\textheight}
    \setlength{\figW}{0.6\textwidth}
    \input{figs/springPendulum/E_diff.tikz}
    \caption{Energy function \eqref{discrete_energy_function} increments}
    \label{fig:springPendulum_energy_function_diff}
\end{figure}

\section{CONCLUSION \& OUTLOOK}\label{sec_conclusion}

The present approach to the numerical integration of dynamical systems is based on Livens principle (cf. \cite{livens_hamiltons_1919}). This variational principle is characterized by combining Lagrangian and Hamiltonian viewpoints on mechanics. Thus, the need for the inversion of the mass matrix is avoided, which allows the simulation of systems with singular mass matrix. Moreover, the presented framework takes account of configuration-dependent mass matrices. A generalized energy function is introduced, which is identical to the total energy and can be linked to the Hamiltonian if the inverse mass matrix exists.

The framework has been extended to mechanical systems subject to holonomic constraints, and can thus be linked to previous works \cite{kinon_ggl_2023, kinon_structure_2023}, where a new variational principle has been introduced. The restrictions therein to constant and non-singular mass matrices have therefore been overcome.

Based on this novel formulation, a structure-preserving discretization has been applied and the emanating time-stepping scheme features both velocity and momentum quantities. As the scheme makes use of Gonzalez discrete gradients (see, e.g. \cite{gonzalez_time_1996}), the proposed integrator discretely covers the conservation of the generalized energy (which differs from the total energy of the system if the mass matrix is configuration-dependent) and aims at the preservation of momentum maps corresponding to the system's symmetries. The integrator satisfies the constraint functions down to computer precision. Moreover, due to the relation with the midpoint rule, the scheme is at most second-order accurate.

The newly proposed structure-preserving integration scheme enhances the method presented by Gonzalez \cite{gonzalez_mechanical_1999} with respect to the formulation in a more general (not necessarily Hamiltonian) framework, which makes possible the energy-consistent simulation of systems with configuration-dependent or singular mass matrices.

The numerical properties of the present method have been demonstrated in representative examples of multibody dynamics: the well-known spring pendulum in spherical coordinates and a mass-spring system with singular mass matrix taken from \cite{udwadia06}.

It should be noted that in the present work we have focussed on the design of an energy-preserving scheme. The development of variational integrators based on Livens principle with configuration dependent or singular mass matrices is interesting, as well. Especially the design of higher order methods by following the ideas presented in Wenger et al. \cite{wenger_construction_2017} might be in the scope of future research . In this connection, the approach recently developed in Altmann \& Herzog \cite{altmann22} should also be applicable.

\section*{ACKNOWLEDGMENTS}
We are thankful for insightful discussions with Felix Zähringer and Moritz Hille. Financial support for this work by the DFG (German Research Foundation) – project numbers 227928419 and 442997215 - is greatfully acknowledged.
\bibliographystyle{splncs} 	
\bibliography{MyBibliography}

\end{document}

%% file: colors.tex
\definecolor{color1}{RGB}{230, 159, 0}
\definecolor{color2}{RGB}{86, 180, 233}
\definecolor{color3}{RGB}{204, 121, 167}
\definecolor{color4}{RGB}{0, 158, 115}
\definecolor{color5}{RGB}{0, 114, 178}
\definecolor{color6}{RGB}{213, 94, 0}
\definecolor{color7}{RGB}{240, 228, 66}
\definecolor{colorblack}{RGB}{0, 0, 0}          

\definecolor{cgrey}{RGB}{128, 128, 128}
\definecolor{cdarkgrey}{RGB}{98, 98, 98}

\colorlet{corange}{color1}
\colorlet{ccyan}{color2}
\colorlet{cviolet}{color3}
\colorlet{cgreen}{color4}
\colorlet{cblue}{color5}
\colorlet{cred}{color6}
\colorlet{cyellow}{color7}

\colorlet{ctop}{color1}
\colorlet{ccenter}{color2}
\colorlet{cbottom}{color3}

\colorlet{cpath1}{color1}
\colorlet{cpath2}{color2}
\colorlet{cpath3}{color3}
\colorlet{cpath4}{color4}
\colorlet{cpath5}{color7}

\colorlet{cpathneutral}{cgrey}

\colorlet{cpathm90}{color1}
\colorlet{cpathm45}{color2}
\colorlet{cpath0}{color3}
\colorlet{cpath45}{color4}
\colorlet{cpath90}{color7}

\colorlet{cmtoa}{cred}
\colorlet{cHS01}{cgreen}
\colorlet{cHS02}{color4}
\colorlet{cHS03}{cblue}
\colorlet{cmtlinear}{cyellow}
\colorlet{cmtlinearcompliance}{cviolet}

%% file: figs/federmasse.pdf_tex
\begingroup%
  \makeatletter%
  \providecommand\color[2][]{%
    \errmessage{(Inkscape) Color is used for the text in Inkscape, but the package 'color.sty' is not loaded}%
    \renewcommand\color[2][]{}%
  }%
  \providecommand\transparent[1]{%
    \errmessage{(Inkscape) Transparency is used (non-zero) for the text in Inkscape, but the package 'transparent.sty' is not loaded}%
    \renewcommand\transparent[1]{}%
  }%
  \providecommand\rotatebox[2]{#2}%
  \newcommand*\fsize{\dimexpr\f@size pt\relax}%
  \newcommand*\lineheight[1]{\fontsize{\fsize}{#1\fsize}\selectfont}%
  \ifx\svgwidth\undefined%
    \setlength{\unitlength}{533.82401226bp}%
    \ifx\svgscale\undefined%
      \relax%
    \else%
      \setlength{\unitlength}{\unitlength * \real{\svgscale}}%
    \fi%
  \else%
    \setlength{\unitlength}{\svgwidth}%
  \fi%
  \global\let\svgwidth\undefined%
  \global\let\svgscale\undefined%
  \makeatother%
  \begin{picture}(1,0.41013836)%
    \lineheight{1}%
    \setlength\tabcolsep{0pt}%
    \put(0,0){\includegraphics[width=\unitlength,page=1]{./figs/federmasse.pdf}}%
    \put(0.18207993,0.32220167){\color[rgb]{0,0,0}\makebox(0,0)[lt]{\lineheight{1.25}\smash{\begin{tabular}[t]{l}$k_1$\end{tabular}}}}%
    \put(0.3926667,0.26677742){\color[rgb]{0,0,0}\makebox(0,0)[lt]{\lineheight{1.25}\smash{\begin{tabular}[t]{l}$m_1$\end{tabular}}}}%
    \put(0,0){\includegraphics[width=\unitlength,page=2]{./figs/federmasse.pdf}}%
    \put(0.39934788,0.36744096){\color[rgb]{0,0,0}\makebox(0,0)[lt]{\lineheight{1.25}\smash{\begin{tabular}[t]{l}$w$\end{tabular}}}}%
    \put(0,0){\includegraphics[width=\unitlength,page=3]{./figs/federmasse.pdf}}%
    \put(0.6581712,0.32220169){\color[rgb]{0,0,0}\makebox(0,0)[lt]{\lineheight{1.25}\smash{\begin{tabular}[t]{l}$k_2$\end{tabular}}}}%
    \put(0.90124393,0.26677742){\color[rgb]{0,0,0}\makebox(0,0)[lt]{\lineheight{1.25}\smash{\begin{tabular}[t]{l}$m_2$\end{tabular}}}}%
    \put(0,0){\includegraphics[width=\unitlength,page=4]{./figs/federmasse.pdf}}%
    \put(0.20802092,0.08646108){\color[rgb]{0,0,0}\makebox(0,0)[lt]{\lineheight{1.25}\smash{\begin{tabular}[t]{l}$x_1+l_{10}$\end{tabular}}}}%
    \put(0,0){\includegraphics[width=\unitlength,page=5]{./figs/federmasse.pdf}}%
    \put(0.42198766,0.01490395){\color[rgb]{0,0,0}\makebox(0,0)[lt]{\lineheight{1.25}\smash{\begin{tabular}[t]{l}$q_2$\end{tabular}}}}%
    \put(0.71502191,0.01490395){\color[rgb]{0,0,0}\makebox(0,0)[lt]{\lineheight{1.25}\smash{\begin{tabular}[t]{l}$x_2+l_{20}$\end{tabular}}}}%
  \end{picture}%
\endgroup%

%% file: figs/massspring/energy.tikz
\begin{tikzpicture}

  \begin{axis}[%
      width=0.951\figW,
      height=\figH,
      at={(0\figW,0\figH)},
      scale only axis,
      xmin=0,
      xmax=10,
      xlabel style={font=\color{white!15!black}},
      xlabel={$t$},
      ymin=-0.1,
      ymax=1.1,
      ylabel style={font=\color{white!15!black}},
      axis background/.style={fill=white},
      legend style={at={(0.98,0.03)}, anchor=south east,legend cell align=left, align=left, draw=white!15!black}
    ]
    \addplot [color=black, line width=1pt, dotted]
    table[row sep=crcr]{%
        0	1\\
        0.1	0.9802\\
        0.2	0.9222\\
        0.3	0.8316\\
        0.4	0.7213\\
        0.5	0.6129\\
        0.6	0.5312\\
        0.7	0.4941\\
        0.8	0.5031\\
        0.9	0.5441\\
        1	0.5959\\
        1.1	0.6391\\
        1.2	0.6602\\
        1.3	0.6488\\
        1.4	0.5955\\
        1.5	0.4932\\
        1.6	0.3464\\
        1.7	0.1834\\
        1.8	0.05514\\
        1.9	0.01207\\
        2	0.06885\\
        2.1	0.193\\
        2.2	0.3293\\
        2.3	0.4345\\
        2.4	0.492\\
        2.5	0.5056\\
        2.6	0.4877\\
        2.7	0.452\\
        2.8	0.413\\
        2.9	0.3868\\
        3	0.3897\\
        3.1	0.4331\\
        3.2	0.5169\\
        3.3	0.6284\\
        3.4	0.7472\\
        3.5	0.853\\
        3.6	0.9319\\
        3.7	0.977\\
        3.8	0.9862\\
        3.9	0.9611\\
        4	0.9061\\
        4.1	0.8303\\
        4.2	0.7487\\
        4.3	0.68\\
        4.4	0.6401\\
        4.5	0.6351\\
        4.6	0.6584\\
        4.7	0.6949\\
        4.8	0.7284\\
        4.9	0.7461\\
        5	0.7395\\
        5.1	0.7025\\
        5.2	0.6295\\
        5.3	0.5187\\
        5.4	0.3779\\
        5.5	0.2314\\
        5.6	0.1167\\
        5.7	0.06672\\
        5.8	0.089\\
        5.9	0.16\\
        6	0.2418\\
        6.1	0.3042\\
        6.2	0.3355\\
        6.3	0.3397\\
        6.4	0.3288\\
        6.5	0.3161\\
        6.6	0.3141\\
        6.7	0.3325\\
        6.8	0.3777\\
        6.9	0.4504\\
        7	0.5449\\
        7.1	0.6505\\
        7.2	0.754\\
        7.3	0.844\\
        7.4	0.9124\\
        7.5	0.9553\\
        7.6	0.9724\\
        7.7	0.9659\\
        7.8	0.94\\
        7.9	0.9012\\
        8	0.8577\\
        8.1	0.8183\\
        8.2	0.7895\\
        8.3	0.7741\\
        8.4	0.7697\\
        8.5	0.7701\\
        8.6	0.7679\\
        8.7	0.7565\\
        8.8	0.7305\\
        8.9	0.6862\\
        9	0.6204\\
        9.1	0.5317\\
        9.2	0.4225\\
        9.3	0.3024\\
        9.4	0.1887\\
        9.5	0.1028\\
        9.6	0.05988\\
        9.7	0.06139\\
        9.8	0.09435\\
        9.9	0.14\\
        10	0.1836\\
      };
    \addlegendentry{$T$}

    \addplot [color=black, line width=1pt, densely dashed]
    table[row sep=crcr]{%
        0	0\\
        0.1	0.01977\\
        0.2	0.07776\\
        0.3	0.1684\\
        0.4	0.2787\\
        0.5	0.3871\\
        0.6	0.4688\\
        0.7	0.5059\\
        0.8	0.4969\\
        0.9	0.4559\\
        1	0.4041\\
        1.1	0.3609\\
        1.2	0.3398\\
        1.3	0.3512\\
        1.4	0.4045\\
        1.5	0.5068\\
        1.6	0.6536\\
        1.7	0.8166\\
        1.8	0.9449\\
        1.9	0.9879\\
        2	0.9311\\
        2.1	0.807\\
        2.2	0.6707\\
        2.3	0.5655\\
        2.4	0.508\\
        2.5	0.4944\\
        2.6	0.5123\\
        2.7	0.548\\
        2.8	0.587\\
        2.9	0.6132\\
        3	0.6103\\
        3.1	0.5669\\
        3.2	0.4831\\
        3.3	0.3716\\
        3.4	0.2528\\
        3.5	0.147\\
        3.6	0.06807\\
        3.7	0.02305\\
        3.8	0.01376\\
        3.9	0.03886\\
        4	0.09389\\
        4.1	0.1697\\
        4.2	0.2513\\
        4.3	0.32\\
        4.4	0.3599\\
        4.5	0.3649\\
        4.6	0.3416\\
        4.7	0.3051\\
        4.8	0.2716\\
        4.9	0.2539\\
        5	0.2605\\
        5.1	0.2975\\
        5.2	0.3705\\
        5.3	0.4813\\
        5.4	0.6221\\
        5.5	0.7686\\
        5.6	0.8833\\
        5.7	0.9333\\
        5.8	0.911\\
        5.9	0.84\\
        6	0.7582\\
        6.1	0.6958\\
        6.2	0.6645\\
        6.3	0.6603\\
        6.4	0.6712\\
        6.5	0.6839\\
        6.6	0.6859\\
        6.7	0.6675\\
        6.8	0.6223\\
        6.9	0.5496\\
        7	0.4551\\
        7.1	0.3495\\
        7.2	0.246\\
        7.3	0.156\\
        7.4	0.08761\\
        7.5	0.04466\\
        7.6	0.02755\\
        7.7	0.03412\\
        7.8	0.06002\\
        7.9	0.09881\\
        8	0.1423\\
        8.1	0.1817\\
        8.2	0.2105\\
        8.3	0.2259\\
        8.4	0.2303\\
        8.5	0.2299\\
        8.6	0.2321\\
        8.7	0.2435\\
        8.8	0.2695\\
        8.9	0.3138\\
        9	0.3796\\
        9.1	0.4683\\
        9.2	0.5775\\
        9.3	0.6976\\
        9.4	0.8113\\
        9.5	0.8972\\
        9.6	0.9401\\
        9.7	0.9386\\
        9.8	0.9056\\
        9.9	0.86\\
        10	0.8164\\
      };
    \addlegendentry{$V$}

    \addplot [color=black, line width=1pt]
    table[row sep=crcr]{%
        0	1\\
        0.1	1\\
        0.2	1\\
        0.3	1\\
        0.4	1\\
        0.5	1\\
        0.6	1\\
        0.7	1\\
        0.8	1\\
        0.9	1\\
        1	1\\
        1.1	1\\
        1.2	1\\
        1.3	1\\
        1.4	1\\
        1.5	1\\
        1.6	1\\
        1.7	1\\
        1.8	1\\
        1.9	1\\
        2	1\\
        2.1	1\\
        2.2	1\\
        2.3	1\\
        2.4	1\\
        2.5	1\\
        2.6	1\\
        2.7	1\\
        2.8	1\\
        2.9	1\\
        3	1\\
        3.1	1\\
        3.2	1\\
        3.3	1\\
        3.4	1\\
        3.5	1\\
        3.6	1\\
        3.7	1\\
        3.8	1\\
        3.9	1\\
        4	1\\
        4.1	1\\
        4.2	1\\
        4.3	1\\
        4.4	1\\
        4.5	1\\
        4.6	1\\
        4.7	1\\
        4.8	1\\
        4.9	1\\
        5	1\\
        5.1	1\\
        5.2	1\\
        5.3	1\\
        5.4	1\\
        5.5	1\\
        5.6	1\\
        5.7	1\\
        5.8	1\\
        5.9	1\\
        6	1\\
        6.1	1\\
        6.2	1\\
        6.3	1\\
        6.4	1\\
        6.5	1\\
        6.6	1\\
        6.7	1\\
        6.8	1\\
        6.9	1\\
        7	1\\
        7.1	1\\
        7.2	1\\
        7.3	1\\
        7.4	1\\
        7.5	1\\
        7.6	1\\
        7.7	1\\
        7.8	1\\
        7.9	1\\
        8	1\\
        8.1	1\\
        8.2	1\\
        8.3	1\\
        8.4	1\\
        8.5	1\\
        8.6	1\\
        8.7	1\\
        8.8	1\\
        8.9	1\\
        9	1\\
        9.1	1\\
        9.2	1\\
        9.3	1\\
        9.4	1\\
        9.5	1\\
        9.6	1\\
        9.7	1\\
        9.8	1\\
        9.9	1\\
        10	1\\
      };
    \addlegendentry{$E$}

  \end{axis}
\end{tikzpicture}%

%% file: figs/massspring/H_diff.tikz
\begin{tikzpicture}

  \begin{axis}[%
    width=0.951\figW,
    height=\figH,
    at={(0\figW,0\figH)},
    scale only axis,
    xmin=0,
    xmax=10,
    xlabel style={font=\color{white!15!black}},
    xlabel={$t$},
    ymin=-2e-15,
    ymax=2e-15,
    ylabel style={font=\color{white!15!black}},
    ylabel={$E^{n+1}-E^{n}$},
    axis background/.style={fill=white}
    ]
    \addplot [color=black, line width=1pt, forget plot]
    table[row sep=crcr]{%
        0	-1.11e-16\\
        0.1	0\\
        0.2	1.11e-16\\
        0.3	-2.22e-16\\
        0.4	2.22e-16\\
        0.5	-1.11e-16\\
        0.6	1.11e-16\\
        0.7	-1.11e-16\\
        0.8	-3.331e-16\\
        0.9	4.441e-16\\
        1	-2.22e-16\\
        1.1	1.11e-16\\
        1.2	1.11e-16\\
        1.3	-5.551e-16\\
        1.4	4.441e-16\\
        1.5	0\\
        1.6	-3.331e-16\\
        1.7	3.331e-16\\
        1.8	3.331e-16\\
        1.9	0\\
        2	-2.22e-16\\
        2.1	2.22e-16\\
        2.2	-4.441e-16\\
        2.3	2.22e-16\\
        2.4	0\\
        2.5	0\\
        2.6	-2.22e-16\\
        2.7	2.22e-16\\
        2.8	0\\
        2.9	0\\
        3	-1.11e-16\\
        3.1	1.11e-16\\
        3.2	0\\
        3.3	0\\
        3.4	2.22e-16\\
        3.5	0\\
        3.6	0\\
        3.7	0\\
        3.8	0\\
        3.9	-2.22e-16\\
        4	0\\
        4.1	4.441e-16\\
        4.2	-2.22e-16\\
        4.3	0\\
        4.4	2.22e-16\\
        4.5	0\\
        4.6	2.22e-16\\
        4.7	0\\
        4.8	2.22e-16\\
        4.9	2.22e-16\\
        5	-2.22e-16\\
        5.1	0\\
        5.2	0\\
        5.3	0\\
        5.4	0\\
        5.5	2.22e-16\\
        5.6	-4.441e-16\\
        5.7	4.441e-16\\
        5.8	-2.22e-16\\
        5.9	0\\
        6	0\\
        6.1	-2.22e-16\\
        6.2	-2.22e-16\\
        6.3	4.441e-16\\
        6.4	-2.22e-16\\
        6.5	2.22e-16\\
        6.6	-2.22e-16\\
        6.7	2.22e-16\\
        6.8	-4.441e-16\\
        6.9	0\\
        7	0\\
        7.1	4.441e-16\\
        7.2	0\\
        7.3	-2.22e-16\\
        7.4	2.22e-16\\
        7.5	-4.441e-16\\
        7.6	2.22e-16\\
        7.7	-4.441e-16\\
        7.8	2.22e-16\\
        7.9	0\\
        8	-2.22e-16\\
        8.1	2.22e-16\\
        8.2	-2.22e-16\\
        8.3	4.441e-16\\
        8.4	0\\
        8.5	-2.22e-16\\
        8.6	2.22e-16\\
        8.7	2.22e-16\\
        8.8	-4.441e-16\\
        8.9	0\\
        9	0\\
        9.1	2.22e-16\\
        9.2	-2.22e-16\\
        9.3	0\\
        9.4	0\\
        9.5	2.22e-16\\
        9.6	-2.22e-16\\
        9.7	4.441e-16\\
        9.8	0\\
        9.9	-2.22e-16\\
      };
    \addplot [color=black, dashed, line width=0.5pt, forget plot]
    table[row sep=crcr]{%
        0	1e-15\\
        10	1e-15\\
      };
    \addplot [color=black, dashed, line width=0.5pt, forget plot]
    table[row sep=crcr]{%
        0	-1e-15\\
        10	-1e-15\\
      };
  \end{axis}
\end{tikzpicture}%

%% file: figs/massspring/g_pos.tikz
\begin{tikzpicture}

  \begin{axis}[%
      width=0.951\figW,
      height=\figH,
      at={(0\figW,0\figH)},
      scale only axis,
      xmin=0,
      xmax=10,
      xlabel style={font=\color{white!15!black}},
      xlabel={$t$},
      ymin=-2e-15,
      ymax=2e-15,
      ylabel style={font=\color{white!15!black}},
      ylabel={$g(\vec{q}\n)$},
      axis background/.style={fill=white}
    ]
    \addplot [color=black, line width=1pt, forget plot]
    table[row sep=crcr]{%
        0	0\\
        0.1	0\\
        0.2	0\\
        0.3	0\\
        0.4	0\\
        0.5	0\\
        0.6	0\\
        0.7	0\\
        0.8	0\\
        0.9	0\\
        1	0\\
        1.1	0\\
        1.2	0\\
        1.3	0\\
        1.4	0\\
        1.5	0\\
        1.6	0\\
        1.7	0\\
        1.8	0\\
        1.9	0\\
        2	0\\
        2.1	0\\
        2.2	0\\
        2.3	0\\
        2.4	0\\
        2.5	0\\
        2.6	0\\
        2.7	0\\
        2.8	0\\
        2.9	0\\
        3	0\\
        3.1	0\\
        3.2	0\\
        3.3	0\\
        3.4	0\\
        3.5	0\\
        3.6	0\\
        3.7	0\\
        3.8	0\\
        3.9	-2.22e-16\\
        4	0\\
        4.1	2.22e-16\\
        4.2	0\\
        4.3	0\\
        4.4	0\\
        4.5	0\\
        4.6	0\\
        4.7	0\\
        4.8	0\\
        4.9	0\\
        5	0\\
        5.1	0\\
        5.2	0\\
        5.3	0\\
        5.4	0\\
        5.5	0\\
        5.6	0\\
        5.7	0\\
        5.8	2.22e-16\\
        5.9	0\\
        6	0\\
        6.1	0\\
        6.2	0\\
        6.3	-2.22e-16\\
        6.4	2.22e-16\\
        6.5	0\\
        6.6	0\\
        6.7	0\\
        6.8	0\\
        6.9	0\\
        7	0\\
        7.1	2.22e-16\\
        7.2	0\\
        7.3	0\\
        7.4	0\\
        7.5	0\\
        7.6	0\\
        7.7	0\\
        7.8	0\\
        7.9	0\\
        8	0\\
        8.1	0\\
        8.2	0\\
        8.3	0\\
        8.4	0\\
        8.5	0\\
        8.6	0\\
        8.7	0\\
        8.8	0\\
        8.9	0\\
        9	0\\
        9.1	0\\
        9.2	0\\
        9.3	0\\
        9.4	0\\
        9.5	0\\
        9.6	0\\
        9.7	0\\
        9.8	0\\
        9.9	0\\
        10	0\\
      };
    \addplot [color=black, dashed, line width=0.5pt, forget plot]
    table[row sep=crcr]{%
        0	1e-15\\
        10	1e-15\\
      };
    \addplot [color=black, dashed, line width=0.5pt, forget plot]
    table[row sep=crcr]{%
        0	-1e-15\\
        10	-1e-15\\
      };
  \end{axis}
\end{tikzpicture}%

%% file: figs/pendulum.pdf_tex
\begingroup%
  \makeatletter%
  \providecommand\color[2][]{%
    \errmessage{(Inkscape) Color is used for the text in Inkscape, but the package 'color.sty' is not loaded}%
    \renewcommand\color[2][]{}%
  }%
  \providecommand\transparent[1]{%
    \errmessage{(Inkscape) Transparency is used (non-zero) for the text in Inkscape, but the package 'transparent.sty' is not loaded}%
    \renewcommand\transparent[1]{}%
  }%
  \providecommand\rotatebox[2]{#2}%
  \newcommand*\fsize{\dimexpr\f@size pt\relax}%
  \newcommand*\lineheight[1]{\fontsize{\fsize}{#1\fsize}\selectfont}%
  \ifx\svgwidth\undefined%
    \setlength{\unitlength}{219.49030869bp}%
    \ifx\svgscale\undefined%
      \relax%
    \else%
      \setlength{\unitlength}{\unitlength * \real{\svgscale}}%
    \fi%
  \else%
    \setlength{\unitlength}{\svgwidth}%
  \fi%
  \global\let\svgwidth\undefined%
  \global\let\svgscale\undefined%
  \makeatother%
  \begin{picture}(1,1.23945887)%
    \lineheight{1}%
    \setlength\tabcolsep{0pt}%
    \put(0,0){\includegraphics[width=\unitlength,page=1]{figs/pendulum.pdf}}%
    \put(0.27294583,0.88718501){\color[rgb]{0,0,0}\makebox(0,0)[lt]{\lineheight{1.25}\smash{\begin{tabular}[t]{l}$\varphi$\end{tabular}}}}%
    \put(0.18153497,0.72512735){\color[rgb]{0,0,0}\makebox(0,0)[lt]{\lineheight{1.25}\smash{\begin{tabular}[t]{l}$\theta$\end{tabular}}}}%
    \put(0.85887014,0.19412804){\color[rgb]{0.5254902,0.5254902,0.5254902}\makebox(0,0)[lt]{\lineheight{1.25}\smash{\begin{tabular}[t]{l}$m$\end{tabular}}}}%
    \put(0.57776142,0.38565296){\color[rgb]{0,0,0}\makebox(0,0)[lt]{\lineheight{1.25}\smash{\begin{tabular}[t]{l}$r$\end{tabular}}}}%
    \put(0,0){\includegraphics[width=\unitlength,page=2]{figs/pendulum.pdf}}%
    \put(0.46004728,0.68651785){\color[rgb]{0.5254902,0.5254902,0.5254902}\makebox(0,0)[lt]{\lineheight{1.25}\smash{\begin{tabular}[t]{l}$EA$\end{tabular}}}}%
    \put(0,0){\includegraphics[width=\unitlength,page=3]{figs/pendulum.pdf}}%
  \end{picture}%
\endgroup%

%% file: figs/springPendulum/energy.tikz

%
\begin{tikzpicture}

  \begin{axis}[%
      width=0.951\figW,
      height=\figH,
      at={(0\figW,0\figH)},
      scale only axis,
      xmin=0,
      xmax=1,
      xlabel style={font=\color{white!15!black}},
      xlabel={$t$},
      ymin=-0.1496,
      ymax=1.646,
      ylabel style={font=\color{white!15!black}},
      axis background/.style={fill=white},
      legend style={at={(0.98,0.03)}, anchor=south east,legend cell align=left, align=left, draw=white!15!black}
    ]
    \addplot [color=black, line width=1pt, dotted]
    table[row sep=crcr]{%
        0	1.103\\
        0.01	1.114\\
        0.02	1.146\\
        0.03	1.195\\
        0.04	1.255\\
        0.05	1.319\\
        0.06	1.38\\
        0.07	1.432\\
        0.08	1.471\\
        0.09	1.492\\
        0.1	1.496\\
        0.11	1.484\\
        0.12	1.459\\
        0.13	1.426\\
        0.14	1.39\\
        0.15	1.357\\
        0.16	1.33\\
        0.17	1.313\\
        0.18	1.309\\
        0.19	1.318\\
        0.2	1.338\\
        0.21	1.368\\
        0.22	1.403\\
        0.23	1.438\\
        0.24	1.469\\
        0.25	1.49\\
        0.26	1.496\\
        0.27	1.486\\
        0.28	1.459\\
        0.29	1.415\\
        0.3	1.359\\
        0.31	1.296\\
        0.32	1.232\\
        0.33	1.175\\
        0.34	1.132\\
        0.35	1.107\\
        0.36	1.104\\
        0.37	1.123\\
        0.38	1.162\\
        0.39	1.215\\
        0.4	1.278\\
        0.41	1.342\\
        0.42	1.4\\
        0.43	1.448\\
        0.44	1.48\\
        0.45	1.495\\
        0.46	1.493\\
        0.47	1.476\\
        0.48	1.448\\
        0.49	1.413\\
        0.5	1.377\\
        0.51	1.346\\
        0.52	1.322\\
        0.53	1.31\\
        0.54	1.311\\
        0.55	1.324\\
        0.56	1.348\\
        0.57	1.381\\
        0.58	1.416\\
        0.59	1.451\\
        0.6	1.478\\
        0.61	1.494\\
        0.62	1.495\\
        0.63	1.478\\
        0.64	1.444\\
        0.65	1.396\\
        0.66	1.336\\
        0.67	1.272\\
        0.68	1.21\\
        0.69	1.157\\
        0.7	1.121\\
        0.71	1.104\\
        0.72	1.109\\
        0.73	1.136\\
        0.74	1.181\\
        0.75	1.239\\
        0.76	1.303\\
        0.77	1.365\\
        0.78	1.42\\
        0.79	1.462\\
        0.8	1.488\\
        0.81	1.496\\
        0.82	1.488\\
        0.83	1.466\\
        0.84	1.435\\
        0.85	1.399\\
        0.86	1.364\\
        0.87	1.335\\
        0.88	1.316\\
        0.89	1.309\\
        0.9	1.315\\
        0.91	1.333\\
        0.92	1.361\\
        0.93	1.395\\
        0.94	1.431\\
        0.95	1.463\\
        0.96	1.486\\
        0.97	1.496\\
        0.98	1.49\\
        0.99	1.466\\
        1	1.426\\
      };
    \addlegendentry{$T$}

    \addplot [color=black, line width=1pt, densely dashed]
    table[row sep=crcr]{%
        0	0.394\\
        0.01	0.3829\\
        0.02	0.3509\\
        0.03	0.3019\\
        0.04	0.2418\\
        0.05	0.1775\\
        0.06	0.1161\\
        0.07	0.0639\\
        0.08	0.02576\\
        0.09	0.004476\\
        0.1	0.0006299\\
        0.11	0.01271\\
        0.12	0.03744\\
        0.13	0.0703\\
        0.14	0.1062\\
        0.15	0.1398\\
        0.16	0.1668\\
        0.17	0.1834\\
        0.18	0.1876\\
        0.19	0.1788\\
        0.2	0.1582\\
        0.21	0.1284\\
        0.22	0.09334\\
        0.23	0.05795\\
        0.24	0.02742\\
        0.25	0.006738\\
        0.26	5.498e-05\\
        0.27	0.01006\\
        0.28	0.03756\\
        0.29	0.08115\\
        0.3	0.1373\\
        0.31	0.2005\\
        0.32	0.2642\\
        0.33	0.321\\
        0.34	0.3645\\
        0.35	0.3893\\
        0.36	0.3924\\
        0.37	0.3735\\
        0.38	0.3348\\
        0.39	0.2811\\
        0.4	0.2187\\
        0.41	0.1546\\
        0.42	0.09584\\
        0.43	0.04828\\
        0.44	0.01603\\
        0.45	0.001104\\
        0.46	0.003309\\
        0.47	0.02044\\
        0.48	0.04869\\
        0.49	0.08326\\
        0.5	0.1189\\
        0.51	0.1506\\
        0.52	0.1741\\
        0.53	0.1864\\
        0.54	0.1858\\
        0.55	0.1724\\
        0.56	0.148\\
        0.57	0.1158\\
        0.58	0.07999\\
        0.59	0.04578\\
        0.6	0.01834\\
        0.61	0.00244\\
        0.62	0.001741\\
        0.63	0.01828\\
        0.64	0.05206\\
        0.65	0.1009\\
        0.66	0.1604\\
        0.67	0.2246\\
        0.68	0.2865\\
        0.69	0.339\\
        0.7	0.3759\\
        0.71	0.3927\\
        0.72	0.3874\\
        0.73	0.3605\\
        0.74	0.3155\\
        0.75	0.2576\\
        0.76	0.1937\\
        0.77	0.1309\\
        0.78	0.07584\\
        0.79	0.03381\\
        0.8	0.008147\\
        0.81	6.122e-06\\
        0.82	0.008374\\
        0.83	0.03037\\
        0.84	0.06171\\
        0.85	0.09734\\
        0.86	0.132\\
        0.87	0.161\\
        0.88	0.1803\\
        0.89	0.1876\\
        0.9	0.1818\\
        0.91	0.1638\\
        0.92	0.1358\\
        0.93	0.1015\\
        0.94	0.06574\\
        0.95	0.03362\\
        0.96	0.0103\\
        0.97	0.0001704\\
        0.98	0.006305\\
        0.99	0.02995\\
        1	0.07022\\
      };
    \addlegendentry{$V$}

    \addplot [color=black, line width=1pt]
    table[row sep=crcr]{%
        0	1.496\\
        0.01	1.496\\
        0.02	1.496\\
        0.03	1.496\\
        0.04	1.496\\
        0.05	1.496\\
        0.06	1.496\\
        0.07	1.496\\
        0.08	1.496\\
        0.09	1.496\\
        0.1	1.496\\
        0.11	1.496\\
        0.12	1.496\\
        0.13	1.496\\
        0.14	1.496\\
        0.15	1.496\\
        0.16	1.496\\
        0.17	1.496\\
        0.18	1.496\\
        0.19	1.496\\
        0.2	1.496\\
        0.21	1.496\\
        0.22	1.496\\
        0.23	1.496\\
        0.24	1.496\\
        0.25	1.496\\
        0.26	1.496\\
        0.27	1.496\\
        0.28	1.496\\
        0.29	1.496\\
        0.3	1.496\\
        0.31	1.496\\
        0.32	1.496\\
        0.33	1.496\\
        0.34	1.496\\
        0.35	1.496\\
        0.36	1.496\\
        0.37	1.496\\
        0.38	1.496\\
        0.39	1.496\\
        0.4	1.496\\
        0.41	1.496\\
        0.42	1.496\\
        0.43	1.496\\
        0.44	1.496\\
        0.45	1.496\\
        0.46	1.496\\
        0.47	1.496\\
        0.48	1.496\\
        0.49	1.496\\
        0.5	1.496\\
        0.51	1.496\\
        0.52	1.496\\
        0.53	1.496\\
        0.54	1.496\\
        0.55	1.496\\
        0.56	1.496\\
        0.57	1.496\\
        0.58	1.496\\
        0.59	1.496\\
        0.6	1.496\\
        0.61	1.496\\
        0.62	1.496\\
        0.63	1.496\\
        0.64	1.496\\
        0.65	1.496\\
        0.66	1.496\\
        0.67	1.496\\
        0.68	1.496\\
        0.69	1.496\\
        0.7	1.496\\
        0.71	1.496\\
        0.72	1.496\\
        0.73	1.496\\
        0.74	1.496\\
        0.75	1.496\\
        0.76	1.496\\
        0.77	1.496\\
        0.78	1.496\\
        0.79	1.496\\
        0.8	1.496\\
        0.81	1.496\\
        0.82	1.496\\
        0.83	1.496\\
        0.84	1.496\\
        0.85	1.496\\
        0.86	1.496\\
        0.87	1.496\\
        0.88	1.496\\
        0.89	1.496\\
        0.9	1.496\\
        0.91	1.496\\
        0.92	1.496\\
        0.93	1.496\\
        0.94	1.496\\
        0.95	1.496\\
        0.96	1.496\\
        0.97	1.496\\
        0.98	1.496\\
        0.99	1.496\\
        1	1.496\\
      };
    \addlegendentry{$E_\mathrm{tot}$}

  \end{axis}
\end{tikzpicture}%

%% file: figs/springPendulum/H_diff.tikz
%
%
\begin{tikzpicture}

  \begin{axis}[%
      width=0.951\figW,
      height=\figH,
      at={(0\figW,0\figH)},
      scale only axis,
      xmin=0,
      xmax=1,
      xlabel style={font=\color{white!15!black}},
      xlabel={$t$},
      ymin=-0.0002,
      ymax=0.0002,
      ylabel style={font=\color{white!15!black}},
      ylabel={$E_\mathrm{tot}^{n+1}-E_\mathrm{tot}^{n}$},
      axis background/.style={fill=white}
    ]
    \addplot [color=black, line width=1pt, forget plot]
    table[row sep=crcr]{%
        0	-2.066e-06\\
        0.01	-1.414e-05\\
        0.02	-1.657e-05\\
        0.03	-2.541e-05\\
        0.04	-2.339e-05\\
        0.05	-2.698e-05\\
        0.06	-1.942e-05\\
        0.07	-1.776e-05\\
        0.08	-5.855e-06\\
        0.09	-1.28e-06\\
        0.1	1.166e-05\\
        0.11	1.518e-05\\
        0.12	2.499e-05\\
        0.13	2.365e-05\\
        0.14	2.748e-05\\
        0.15	1.985e-05\\
        0.16	1.792e-05\\
        0.17	5.931e-06\\
        0.18	1.605e-06\\
        0.19	-1.046e-05\\
        0.2	-1.26e-05\\
        0.21	-2.055e-05\\
        0.22	-1.732e-05\\
        0.23	-1.942e-05\\
        0.24	-1.078e-05\\
        0.25	-8.618e-06\\
        0.26	2.556e-06\\
        0.27	5.301e-06\\
        0.28	1.507e-05\\
        0.29	1.474e-05\\
        0.3	2.016e-05\\
        0.31	1.481e-05\\
        0.32	1.506e-05\\
        0.33	5.012e-06\\
        0.34	1.466e-06\\
        0.35	-1.101e-05\\
        0.36	-1.542e-05\\
        0.37	-2.701e-05\\
        0.38	-2.89e-05\\
        0.39	-3.646e-05\\
        0.4	-3.321e-05\\
        0.41	-3.488e-05\\
        0.42	-2.562e-05\\
        0.43	-2.167e-05\\
        0.44	-7.888e-06\\
        0.45	-9.704e-07\\
        0.46	1.371e-05\\
        0.47	1.934e-05\\
        0.48	3.052e-05\\
        0.49	3.091e-05\\
        0.5	3.562e-05\\
        0.51	2.927e-05\\
        0.52	2.771e-05\\
        0.53	1.654e-05\\
        0.54	1.207e-05\\
        0.55	3.507e-07\\
        0.56	-2.486e-06\\
        0.57	-1.062e-05\\
        0.58	-8.743e-06\\
        0.59	-1.167e-05\\
        0.6	-5.07e-06\\
        0.61	-4.307e-06\\
        0.62	4.294e-06\\
        0.63	5.337e-06\\
        0.64	1.237e-05\\
        0.65	1.046e-05\\
        0.66	1.342e-05\\
        0.67	7.042e-06\\
        0.68	5.518e-06\\
        0.69	-4.632e-06\\
        0.7	-8.962e-06\\
        0.71	-2.041e-05\\
        0.72	-2.451e-05\\
        0.73	-3.393e-05\\
        0.74	-3.452e-05\\
        0.75	-3.895e-05\\
        0.76	-3.368e-05\\
        0.77	-3.166e-05\\
        0.78	-2.018e-05\\
        0.79	-1.265e-05\\
        0.8	2.748e-06\\
        0.81	1.227e-05\\
        0.82	2.706e-05\\
        0.83	3.353e-05\\
        0.84	4.271e-05\\
        0.85	4.184e-05\\
        0.86	4.252e-05\\
        0.87	3.337e-05\\
        0.88	2.675e-05\\
        0.89	1.259e-05\\
        0.9	3.614e-06\\
        0.91	-9.677e-06\\
        0.92	-1.512e-05\\
        0.93	-2.253e-05\\
        0.94	-2.102e-05\\
        0.95	-2.114e-05\\
        0.96	-1.335e-05\\
        0.97	-8.645e-06\\
        0.98	1.697e-06\\
        0.99	6.881e-06\\
      };
    \addplot [color=black, dashed, line width=0.5pt, forget plot]
    table[row sep=crcr]{%
        0	0.0001\\
        1	0.0001\\
      };
    \addplot [color=black, dashed, line width=0.5pt, forget plot]
    table[row sep=crcr]{%
        0	-0.0001\\
        1	-0.0001\\
      };
  \end{axis}
\end{tikzpicture}%

%% file: figs/springPendulum/E_diff.tikz
%
%
\begin{tikzpicture}

  \begin{axis}[%
    width=0.951\figW,
    height=\figH,
    at={(0\figW,0\figH)},
    scale only axis,
    xmin=0,
    xmax=1,
    xlabel style={font=\color{white!15!black}},
    xlabel={$t$},
    ymin=-2e-14,
    ymax=2e-14,
    ylabel style={font=\color{white!15!black}},
    ylabel={$E^{n+1}-E^{n}$},
    axis background/.style={fill=white}
    ]
    \addplot [color=black, line width=1pt, forget plot]
    table[row sep=crcr]{%
        0	-2.22e-16\\
        0.01	1.776e-15\\
        0.02	1.332e-15\\
        0.03	4.441e-16\\
        0.04	8.882e-16\\
        0.05	2.22e-16\\
        0.06	0\\
        0.07	0\\
        0.08	2.22e-16\\
        0.09	-2.22e-16\\
        0.1	-2.22e-16\\
        0.11	-2.22e-16\\
        0.12	0\\
        0.13	-2.22e-16\\
        0.14	2.22e-16\\
        0.15	-6.661e-16\\
        0.16	0\\
        0.17	0\\
        0.18	-2.22e-16\\
        0.19	-1.11e-15\\
        0.2	0\\
        0.21	4.441e-16\\
        0.22	-8.882e-16\\
        0.23	4.441e-16\\
        0.24	0\\
        0.25	-6.661e-16\\
        0.26	4.441e-16\\
        0.27	-2.22e-16\\
        0.28	2.22e-16\\
        0.29	-1.11e-15\\
        0.3	1.11e-15\\
        0.31	8.882e-16\\
        0.32	-6.661e-16\\
        0.33	1.11e-15\\
        0.34	-4.441e-16\\
        0.35	-2.22e-16\\
        0.36	-1.11e-15\\
        0.37	-1.554e-15\\
        0.38	-4.441e-16\\
        0.39	-4.441e-16\\
        0.4	6.661e-16\\
        0.41	0\\
        0.42	2.22e-16\\
        0.43	2.22e-16\\
        0.44	-2.22e-16\\
        0.45	8.882e-16\\
        0.46	-8.882e-16\\
        0.47	2.22e-16\\
        0.48	-1.332e-15\\
        0.49	2.22e-16\\
        0.5	-2.22e-16\\
        0.51	-2.22e-16\\
        0.52	-6.661e-16\\
        0.53	6.661e-16\\
        0.54	8.882e-16\\
        0.55	-4.441e-16\\
        0.56	6.661e-16\\
        0.57	4.441e-16\\
        0.58	-6.661e-16\\
        0.59	4.441e-16\\
        0.6	-6.661e-16\\
        0.61	4.441e-16\\
        0.62	6.661e-16\\
        0.63	-4.441e-16\\
        0.64	-4.441e-16\\
        0.65	1.776e-15\\
        0.66	2.22e-16\\
        0.67	-1.776e-15\\
        0.68	8.882e-16\\
        0.69	-1.998e-15\\
        0.7	1.332e-15\\
        0.71	1.554e-15\\
        0.72	8.882e-16\\
        0.73	-6.661e-16\\
        0.74	2.22e-16\\
        0.75	0\\
        0.76	-2.22e-16\\
        0.77	-2.22e-16\\
        0.78	0\\
        0.79	8.882e-16\\
        0.8	2.22e-16\\
        0.81	6.661e-16\\
        0.82	2.22e-16\\
        0.83	-2.22e-16\\
        0.84	1.11e-15\\
        0.85	-8.882e-16\\
        0.86	4.441e-16\\
        0.87	-2.22e-16\\
        0.88	-4.441e-16\\
        0.89	4.441e-16\\
        0.9	-2.22e-16\\
        0.91	-4.441e-16\\
        0.92	1.332e-15\\
        0.93	2.22e-16\\
        0.94	-4.441e-16\\
        0.95	-2.22e-16\\
        0.96	2.22e-16\\
        0.97	2.22e-16\\
        0.98	1.11e-15\\
        0.99	-4.441e-16\\
      };
    \addplot [color=black, dashed, line width=0.5pt, forget plot]
    table[row sep=crcr]{%
        0	1e-14\\
        1	1e-14\\
      };
    \addplot [color=black, dashed, line width=0.5pt, forget plot]
    table[row sep=crcr]{%
        0	-1e-14\\
        1	-1e-14\\
      };
  \end{axis}
\end{tikzpicture}%